\documentclass[12pt,reqno, letterpaper]{amsart}
\usepackage{amssymb,epsf}
\usepackage{amstext,amsmath,amsfonts,amscd,amsthm,graphicx}
\usepackage{epsfig}
\usepackage{eufrak}
\usepackage{latexsym}
\usepackage{eucal}
\usepackage{mathtools}

\DeclarePairedDelimiter\floor{\lfloor}{\rfloor}
\theoremstyle{amsart}
\newfont{\fnt}{cmsy10}
\newfont{\sss}{cmr10}
\newfont{\azb}{wncyr10}
\newfont{\azbit}{wncyi10}
\theoremstyle{definition}

\theoremstyle{plain}

\newtheorem{vt}{Theorem}
\newtheorem*{yosh}{Yoshida's Conjecture}
\newtheorem{lm}{Lemma}
\newtheorem{ds}{Corollary}
\newtheorem{prp}{Proposition}
\theoremstyle{definition}
\newtheorem{pz}{Remark}

\newcommand{\diagdots}[3][-25]{%
  \rotatebox{#1}{\makebox[0pt]{\makebox[#2]{\xleaders\hbox{$\cdot$\hskip#3}\hfill\kern0pt}}}%
}
\begin{document}
\title[Towards the Proof of Yoshida's Conjecture]{
{\protect\vspace*{-1.6cm}}
Towards the Proof of Yoshida's Conjecture}
\author{Ji\v{r}\'{i} Jahn} \address{Mathematical Institute, Silesian University in Opava, Na Rybn\'{i}\v{c}ku~1, 746~01 Opava, Czech Republic} \email{Jiri.Jahn@math.slu.cz} \author{Ji\v{r}ina Jahnov\'{a}} \address{Mathematical Institute, Silesian University in Opava, Na Rybn\'{i}\v{c}ku~1, 746~01 Opava, Czech Republic} \email{Jirina.Vodova@math.slu.cz, Jirina.Jahnova@math.slu.cz}
%\address{Mathematical Institute, Silesian University in Opava, Na Rybn\'{i}\v{c}ku~1, 746~01 Opava, Czech Republic}
%\email{Jirina.Vodova@math.slu.cz, Jirina.Jahnova@math.slu.cz}
\begin{abstract} %{\protect\vspace*{-0.7cm}}
\looseness=-1
Yoshida's Conjecture formulated by H. Yoshida in 1989 states that in $\mathbb{C}^{2N}$ equipped with the canonical symplectic form $\mathrm{d}\mathbf{p} \wedge \mathrm{d} \mathbf{q},$ the Hamiltonian flow corresponding to the Hamiltonian \begin{equation*}
H = \frac{1}{2}%\sum_{i=1}^{N} p_{i}^{2}
\sum_{i=1}^{N} p_{i}^{2} + \sum_{i=0}^{N} (q_i - q_{i+1})^k, \ \quad \text{with\ } q_0 = q_{N+1} = 0, \label{system}
\end{equation*}
where $N\geq 3$ is odd and $k\geq 4$ is even, has no global complex meromorphic first integral functionally independent of $H$. For $N=3$ and $N=5$ with $k\geq 4$ arbitrary even number, the result was proved true by Maciejewski, Przybylska and Yoshida in 2012 by means of differential Galois theory. However, the question whether Yoshida's conjecture is true in general, remained open. In this paper we give a proof that this conjecture is in fact true for infinitely many values of $N$ using the results of R. D. Costin which are based on the so-called poly-Painlev\'{e} method devised by M. Kruskal.
\end{abstract}
\maketitle
%\protect\vspace*{-5mm}
\section{Introduction}
General Hamiltonian dynamical systems are traditionally of great interest for both mathematicians and physicists simply because they describe the evolution of many physical systems like a planetary system, electron in an electromagnetic field etc. 
%Over the last years, Hamiltonian systems in the complex domain (the phase space is a subspace of $\mathbb{C}^{2N}$ and the time variable is complex) are given a great amount of attention both in mathematics and physics as seen for example in the theory of complex mechanics and other fields. 
Since the number of independent first integrals determines the dimension of the space that the trajectories fill and therefore provides some information about the possible onset of chaos in the phase space, it is no surprise that one of the basic problems concerning these systems is to decide whether a given system is integrable or not in the sense of Liouville theorem (cf.~ \cite{arnold}), which means to decide whether there are sufficiently many functionally independent first integrals in involution that are of a given class or not. %The question of nonintegrability of Hamiltonian systems of $N$ degrees of freedom is of equally prime interest both in the real and complex domain. Results of this type are in fact mostly expressed in terms of necessary conditions of integrability or in terms of necessary conditions of the existence of $m$ independent integrals in involution, $1 \leq m < N$ (partial integrability).  
In the realm of integrability theory, it is a standard trick of the trade already since the days of S. Kowalevski to consider complexified systems instead of the original ones (see e.g. \cite{maciejewski, yoshida2, morales-ruiz4, morales-ruiz1, morales-ruiz2, morales-ruiz3, yoshida3, yoshida1, ziglin1, ziglin2}) and virtually the same route is taken also in this paper - here, we consider Hamiltonian systems in the complex domain (the phase space is a subspace of $\mathbb{C}^{2N}$ and the time variable is complex).
In the 80's, Ziglin's theory of complex nonintegrability of these complexified systems, in which the necessary conditions of integrability are expressed in terms of certain monodromy matrices of the solutions of the corresponding variational equations along a suitably chosen particular solution of the original Hamiltonian system, started to develop \cite{ziglin1}, \cite{ziglin2}. This theory was further extended by H. Yoshida \cite{yoshida3}, \cite{yoshida1}. %who expressed the monodromy matrices from Ziglin's theory as monodromy matrices of Gauss's hypergeometric equation, albeit only in the case of Hamiltonians with homogeneous potential \cite{yoshida3}, \cite{yoshida1} and thus formulated the necessary conditions of integrability of such a system by means of a system of conditions imposed on the so called integrability coefficient of that Hamiltonian system.
Building on Ziglin's theory, Morales-Ruiz and Ramis were able to obtain further results using the differential Galois theory that relate the integrability of Hamiltonian systems to the solvability of their variational equations around a particular solution (see \cite{morales-ruiz4}, \cite{morales-ruiz1}, \cite{morales-ruiz2} and \cite{morales-ruiz3}).
%A later improvement on Ziglin's theory was achieved by Morales-Ruiz and Ramis in \cite{morales-ruiz4}, \cite{morales-ruiz1}, \cite{morales-ruiz2} and \cite{morales-ruiz3} using the differential Galois theory relating the integrability of Hamiltonian systems to the solvability of their variational equations around a particular solution. 
The necessary conditions of partial integrability of Hamiltonian systems by means of differential Galois theory have been recently studied by Maciejewski, Przybylska and Yoshida \cite{maciejewski}, \cite{yoshida2}.

Although the theory of (non)integrability is developing, there are still systems that seem not to be amenable to analysis using this theory. One of these systems is the N-degrees-of-freedom Hamiltonian system 
\begin{equation}\label{system1}
\frac{\mathrm{d} \bf{q}}{\mathrm{d}t}=\frac{\partial H}{\partial \mathbf{p}}, \qquad \frac{\mathrm{d} \mathbf{p}}{\mathrm{d}t}=-\frac{\partial H}{\partial \mathbf{q}}, \qquad \mathbf{q}=(q_1,\dots,q_N)\in\mathbb{C}^N, \mathbf{p}=(p_1,\dots,p_N)\in\mathbb{C}^N, 
\end{equation}
whose dynamics being given by the Hamiltonian
\begin{equation}
H = \frac{1}{2}%\sum_{i=1}^{N} p_{i}^{2}
\sum_{i=1}^{N} p_{i}^{2} + V_{N,k}(\mathbf{q}), \label{system}
\end{equation}
where the potential is the following homogeneous function
\begin{equation}V_{N,k}(\mathbf{q}) = \sum_{i=0}^{N} (q_i - q_{i+1})^k, \ \quad \text{with\ } q_0 = q_{N+1} = 0,\label{potential}
\end{equation}
and the time $t$ is considered to be a complex variable.

%In this paper we consider a system defined by the Hamiltonian function 

To the best of our knowledge, this system was first considered by H. Yoshida in \cite{yoshida1} where it was shown it does not admit a single additional first integral for $k=4$ and $N=3,5$ except the Hamiltonian itself. In this paper Yoshida formulated the following conjecture
%In this paper we give a proof that the following conjecture due to H. Yoshida \cite{yoshida1, yoshida2} is in fact true:

\begin{yosh}
For an arbitrary odd $N \geq 3,$ and for arbitrary even $k\geq 4,$ the Hamiltonian system with $N$ degrees of freedom, given by the Hamiltonian (\ref{system}) with the potential (\ref{potential}) 
%\begin{equation}
%H = \frac{1}{2}%\sum_{i=1}^{N} p_{i}^{2}
%\mathbf{p}^2 + V_{N,k}(\mathbf{q}),
%\end{equation}
%where 
%\begin{equation}
%\frac12\mathbf{p}^2 = \sum_{i=1}^{N} p_{i}^{2}, \quad V_{N,k}(\mathbf{q}) = \sum_{i=0}^{N} (q_i - q_{i+1})^k, \ \quad \text{with\ } q_0 = q_{n+1} = 0, 
%\end{equation}
does not admit an additional global complex meromorphic first integral. 
\end{yosh}

Yoshida himself actually proved that the conjecture is in fact a corollary to another conjecture of him which he also stated in \cite{yoshida1} in the following form:\\

\noindent \textit{For an arbitrary odd $N\geq 3$ and even $k\geq 4$, the numbers $\Delta_1, \dots, \Delta_N$ given by
\begin{equation}\label{yoshidasuff}\Delta_j=\frac{1}{2k}\sqrt{(k-2)^2+8k\lambda_j} \quad 1\leq j\leq N,\end{equation}
where $\lambda_j:=2(k-1)\sin^2((\pi j)/(2(N+1)))$, are independent over $\mathbb{Q}$. 
}\\

However, in 2007 K. Yoshimura proved that this latter conjecture is not true when he found a counterexample to the case $N=5, k=16$ (for details see \cite{yoshimura, yoshida2}), so that the question whether Yoshida's Conjecture is true or not remained open. In 2012, Maciejewski, Przybylska and Yoshida \cite{yoshida2} proved that Yoshida's Conjecture is true for $N=3,5$ and arbitrary even $k\geq 4$ using the results obtained in \cite{yoshida2} that are based on an application of differential Galois theory to variational equations along a particular solution, and they strongly conjectured that Yoshida's Conjecture is true for an arbitrary odd $N\geq 3$ and even $k \geq 4$.  Unfortunately, the approach through differential Galois theory seems to be virtually impossible to be applied to systems that contain parameters, so that the problem of finding the general proof of Yoshida's Conjecture is still open.

The principal aim of this paper is to give a proof of Yoshida's Conjecture for infinitely many values of $N,$ more precisely for those $N \geq 3$ such that $N \equiv 1 \!\! \mod{6}$ or $N \equiv 3 \!\! \mod{6}$ with an arbitrary even $k \geq 4$. To this end we employ results due to R. D. Costin which are of slightly different nature than the results discussed above. They are based on the ideas of Martin D. Kruskal's poly-Painlev\'{e} method \cite{kruskal4}, \cite{kruskal1}, \cite{kruskal2}, \cite{kruskal3} which, although at times somewhat heuristic in spirit, finds its completely satisfactory and rigorous counterpart in the work of Costin. These results first appeared in \cite{costin1} and their generalization was later given in the Ph.D. thesis of Costin (see \cite{CostinThesis}).

The present paper is organized in the following way:
the main ideas and results pertaining to the proof are described in Section \ref{sekce_costin}. To streamline the proof, we briefly recall several basic facts about Chebyshev polynomials in Section \ref{sekce_chebychev}; these will be later used in the proof. The very proof of Yoshida's Conjecture for the special values of $N$ cited above occupies the rest of the paper and is to be found in Section \ref{Main_results}.

From now on, the term ``(non)integrability" stands for ``complex (non)integrability", and ``meromorphic function" means ``complex meromorphic function".

\section{The poly-Painlev\'{e} test: a review of results due to R. Costin}\label{sekce_costin}
In this section, we give a short account of the results related to polynomial homogeneous systems of second-order ordinary differential equations introduced in \cite{CostinThesis} that are a generalization of similar results obtained in \cite{costin1} in the case of the homogeneous H\'{e}non-Heiles system and at the same time they are elaboration of the poly-Painlev\'{e} method devised by M. D. Kruskal. 

Consider the system of second-order differential equations
\begin{equation}\label{rovnice}
\ddot{q}_m = P_m(\mathbf{q}),\qquad m=1,\dots,N, \qquad \mathbf{q}=(q_1,\dots,q_N), \ t\in\mathbb{C}
\end{equation}
with $P_m$ homogeneous polynomials of degree $k-1.$ Particular solutions $q_m(t)$ of (\ref{rovnice}) that are of the form $q_m(t) = \alpha_m \phi(t),$ $m = 1, \ldots, N,$ $\alpha_m \in \mathbb{C},$ satisfy the conditions
\begin{equation}
\ddot{\phi} = \phi^{k-1}, \quad \alpha_m = P_m(\alpha_1, \ldots, \alpha_N) \label{podminka}
\end{equation}
for every $m = 1, \ldots, N.$ Let $\pmb{\alpha} = (\alpha_1, \ldots, \alpha_N)$ be a nonzero solution of (\ref{podminka}) such that $\alpha_N \neq 0.$ If we introduce new variables $u_1, \ldots, u_{N-1}$ and $Q,$ a small parameter $\epsilon$ by the formulas
\begin{eqnarray*}
q_m(t) & = & \frac{\alpha_m}{\alpha_N} Q(t) + \epsilon u_m(t), \quad m = 1, \ldots, N-1, \\
q_N(t) & = & Q(t)
\end{eqnarray*}
and if we Taylor expand the polynomials $P_m,$ then using homogeneity, the system (\ref{rovnice}) takes the form
\begin{eqnarray}
\ddot{\mathbf{u}} &=& Q^{k-2} \alpha_{N}^{-k+2}M_P \mathbf{u} + \epsilon R(\mathbf{u}, Q, \epsilon), \label{transformed1}\\
\ddot{Q}&=&Q^{k-1}\alpha_N^{-k+1}P_N(\alpha)+\varepsilon R_N(\mathbf{u},Q,\varepsilon) \label{transformed2}
\end{eqnarray}
where $\mathbf{u} = (u_1, \ldots, u_{N-1}),$ $R$ is a vector of polynomials, $R_N$ is a polynomial and $M_P$ is a matrix, given by the formula
\begin{equation}
M_{P} = \left(\frac{\partial P_i}{\partial q_j}(\pmb{\alpha}) - \frac{\alpha_i}{\alpha_N} \frac{\partial P_N}{\partial q_j} (\pmb{\alpha}) \right)_{\!\!ij}, \quad i,j = 1, 2, \ldots, N-1. \label{maticaobecne}
\end{equation}
If the matrix $M_P$ has $N-1$ distinct eigenvalues $\lambda_1, \ldots, \lambda_{N-1},$ then the reduced system (i.e.~ the system corresponding to equations (\ref{transformed1}), (\ref{transformed2}) with $\varepsilon = 0$) can be converted, after integrating once the reduced form of equation (\ref{transformed2}) and eliminating time by treating $Q$ as an independent variable and rescaling $Q,$ into the decoupled system of $N-1$ generalized Lam\'{e} equations of the form
\begin{equation}
(x^k - 1) \frac{\mathrm{d}^{2} v_m}{\mathrm{d} x^2} + \frac{k}{2} x^{k-1} \frac{\mathrm{d} v_m}{\mathrm{d} x} - \frac{k}{2} \lambda_m x^{k-2} v_m = 0, \label{lame}
\end{equation}
for $m = 1, \ldots, N-1.$

In \cite{CostinThesis}, it is proved that $r$ functionally independent first integrals $F_1, \ldots, F_r$ of the system (\ref{rovnice}) defined on a domain $\mathcal{D}=D\times D_{q_N} \subset \mathbb{C}^{2N}$ such that 1) the projection $D_{q_N}$ of $\mathcal{D}$ on the $q_N$-coordinate contains closed paths around the roots of the polynomial $q^k - (Ck)/2 a_{N}^{-k+2}$ (as a polynomial in $q$) for some constant $C\in\mathbb{C}$ and 2) all $F_r$'s are meromorphic along the linear manifold $q_m = \frac{\alpha_m}{\alpha_N} q_N,$ $\dot{q}_m = \frac{\alpha_m}{\alpha_N} \dot{q}_N,$ $m = 1, \ldots, N-1,$ give rise to $r$ independent first integrals of the system (\ref{lame}) that are moreover holomorphic on a domain $\Omega$ whose projection on the $x$-coordinate contains closed paths around all the $k$-th roots of unity.

At this point, employing the results concerning such first integrals of the system (\ref{lame}), Costin was able to prove the following theorem:

\begin{vt}[R. Costin, \cite{CostinThesis}] \label{Costin theorem}
Suppose the matrix $M_P$ has $N-1$ distinct eigenvalues $\lambda_1, \ldots, \lambda_{N-1}$. Let $n$ numbers among the numbers
\begin{equation}
\nu_m=\sqrt{\frac{(k-2)^2}{16}+\frac{k}{2}\lambda_m}\qquad m=1,\dots,N-1 \label{nyem}
\end{equation} 
be irrational. Then the system (\ref{rovnice}) has at most $2N-1-2n$ independent first integrals which are meromorphic near the linear manifold $q_m = \frac{\alpha_m}{\alpha_N} q_N,$ $\dot{q}_m = \frac{\alpha_m}{\alpha_N} \dot{q}_N,$ $m = 1, \ldots, N-1.$
\end{vt}

\section{Chebyshev polynomials of first and second kind: a brief review}\label{sekce_chebychev}
In this section, we briefly recall several rudimentary facts concerning the Chebyshev polynomials of the first and second kind. These technical preliminaries will be used in Section \ref{Main_results} where the certain cases of Yoshida's Conjecture are proved. A general reference for this section is \cite{mason} or \cite{snyder}. If not specified otherwise, every occurrence of the symbols $n,m$ in this section refers to arbitrary nonnegative integers $n$ and $m$.

One of the possibilities to define the \textit{Chebyshev polynomial of the first kind} $T_n(x)$ is to require that it is the unique solution of the recurrence relation 
\begin{equation} T_{n+1}(x)=2xT_n(x)-T_{n-1}(x) \quad (n\geq 1), \label{chebyshev1}\end{equation}
with initial conditions $T_0(x)=1,$ $T_1(x)=x.$
The Chebyshev polynomials of the first kind satisfy the following functional equation, called the \textit{nesting property}:
\begin{equation}
T_n \left(T_m(x) \right) = T_{nm}(x). \label{nesting}
\end{equation}
%Note that the property (\ref{nesting}) uniquely determines the Chebyshev polynomials of the first kind under the condition that $T_n(x)$ is a polynomial of degree $n.$
The \textit{Chebyshev polynomial of the second kind} $U_n(x)$ may be uniquely defined as solutions of the recurrence relation
\begin{equation}
U_{n+1}(x) = 2xU_{n}(x) - U_{n-1}(x) \quad (n \geq 1) \label{chebyshev2}
\end{equation}
with initial conditions $U_0(x)=1,$ $U_{1}(x) = 2x.$
The following representation for the Chebyshev polynomial of the second kind is standard:
\begin{equation}\label{Unsuma}U_n(x)=\sum_{k=0}^{\floor{\frac{n}{2}}}(-1)^k{{n-k}\choose{k}}(2x)^{n-2k},\ n\geq 1.\end{equation}
Next we list a handful of properties that show the mutual connection between the Chebyshev polynomials of the first and second kind and that will turn out to be useful in the sequel:
\begin{equation}
T_{n}(x) = \frac12\left(U_n(x) - U_{n-2}(x) \right) \quad \text{for every } n \geq 2, \label{vlastnost1}
\end{equation}
\begin{equation}
U_n(x) = 2 \sum_{\substack{j = 1 \\ j \text{\ odd}}}^{n} T_j(x), \quad \text{if } n \text{ is odd,} \label{vlastnost2}
\end{equation}
\begin{equation}
U_n(x) = 2 \sum_{\substack{j = 0 \\ j \text{\ even}}}^{n} T_j(x) - 1, \quad \text{if } n \text{ is even.} \label{vlastnost3}
\end{equation}

Finally, we shall later need the following transformation property between the Chebyshev polynomials of the first and second kind:
\begin{equation}
U_{mn - 1}(x) = U_{m-1}\left(T_n(x) \right)U_{n-1}(x), \quad (m,n \geq 1). \label{transform}
\end{equation}

Another important sort of problems, especially when dealing with polynomials, is the question of the location and nature of their roots. The Chebyshev polynomials of either kind are a well-explored area in this respect, since both $T_n(x)$ and $U_n(x)$ have exactly $n$ distinct roots lying in the interval $(-1,1)$ that are given by the formulas
\begin{equation}
x_{k}^{T} = \cos{\left( \frac{\pi (2k - 1)}{2n} \right),} \quad x_{k}^{U} = \cos{\left(\frac{k \pi}{n+1} \right)} \label{koreny}
\end{equation}
for $k = 1, \ldots, n.$

\section{Main results} \label{Main_results}

Consider the Hamiltonian system given by (\ref{system}). It can be shown that the system (\ref{system}) can be equivalently written as a second-order system of the following form:
\begin{equation}
\frac{\mathrm{d}^2 \mathbf{q}}{\mathrm{d} t^2} = - \nabla V_{N,k}(\mathbf{q}) \label{rovnice special},
\end{equation}
where $\nabla$ stands for the usual gradient operator in cartesian coordinates. We note that $- \nabla V_{N,k}(\mathbf{q})$ is a homogeneous polynomial of degree $k-1,$ since $V_{N,k}$ is a homogeneous polynomial of degree $k$ in $\mathbf{q} = (q_1, \ldots, q_N).$ Hence the results from Section \ref{sekce_costin} are applicable in this case.

The corresponding point $\pmb{\alpha} = (\alpha_1, \ldots, \alpha_N)$ that describes the particular solutions of the form $q_{m}(t) = \alpha_{m} \phi(t)$ for the system (\ref{rovnice special}) is any solution of the algebraic system
\begin{equation}
\left( \begin{array}{c}
- k \alpha_1^{k-1} - k (\alpha_1 - \alpha_2)^{k-1} \\
k (\alpha_1 - \alpha_2)^{k-1} - k(\alpha_2 - \alpha_3)^{k-1} \\
\vdots \\
k(\alpha_{N-1} - \alpha_N)^{k-1} - k(\alpha_N)^{k-1}
\end{array} \right) = \left( \begin{array}{c} \alpha_1 \\ \alpha_2 \\ \vdots \\ \alpha_N \end{array} \right). \label{systalg}
\end{equation}
Using the ansatz $\pmb{\alpha} = (\alpha_1, \ldots, \alpha_N) = (\alpha, 0, -\alpha, 0, \ldots, (-1)^{\frac{N - 1}{2}} \alpha),$ we obtain one particular solution $\pmb{\alpha}$ of the system (\ref{systalg}) with $\alpha = (-1)^{\frac{1}{k-2}} (2k)^{\frac{1}{2-k}}, $ where $(-1)^{\frac{1}{k-2}}$ is in general complex-valued, taken with principal values of the argument. 

In this particular case, the $(N-1) \times (N-1)$ matrix $M_P$ from (\ref{maticaobecne}) corresponding to $P = - \nabla V_{N,k}$ takes the following form:
\begin{equation}
M_{P} =  \frac12(k-1) \left( \begin{array}{cccccc} 
2 & -1 & 0 & 0 & \ldots & 0 \\
-1 & 2 & -1 & 0 & \ldots & 0 \\
0 & -1 & 2 & -1 & \ldots & 0 \\
0 & 0 & -1 & 2 & \ddots & \vdots \\
\vdots & \vdots & \vdots & \ddots & \ddots & -1 \\
0 & 0 & 0 & \ldots & -1 & 2 
\end{array} \right) + \frac12(k-1) \left( \begin{array}{cccccc} 
0 & 0 & 0 & \ldots & 0 & a_1 \\
0 & 0 & 0 & \ldots & 0 & a_2 \\
0 & 0 & 0 & \ldots & 0 & a_3 \\
0 & 0 & 0 & \ddots & \vdots & \vdots \\
\vdots & \vdots & \vdots &  & 0 & a_{N-2} \\
0 & 0 & 0 & \ldots & 0 & a_{N-1} 
\end{array} \right), \label{naseMP}
\end{equation}
where
%\begin{equation*}
%M_1 = \frac12(k-1) \left( \begin{array}{cccccc} 
%2 & -1 & 0 & 0 & \ldots & 0 \\
%-1 & 2 & -1 & 0 & \ldots & 0 \\
%0 & -1 & 2 & -1 & \ldots & 0 \\
%0 & 0 & -1 & 2 & \ddots & \vdots \\
%\vdots & \vdots & \vdots & \ddots & \ddots & -1 \\
%0 & 0 & 0 & \ldots & -1 & 2 
%\end{array} \right),
%\end{equation*}
%
%\begin{equation*}
%M_2 = \frac12(k-1) \left( \begin{array}{cccccc} 
%0 & 0 & 0 & \ldots & 0 & a_1 \\
%0 & 0 & 0 & \ldots & 0 & a_2 \\
%0 & 0 & 0 & \ldots & 0 & a_3 \\
%0 & 0 & 0 & \ddots & \vdots & \vdots \\
%\vdots & \vdots & \vdots &  & 0 & a_{N-2} \\
%0 & 0 & 0 & \ldots & 0 & a_{N-1} 
%\end{array} \right)
%\end{equation*}
the numbers $a_i$ read like this:

\begin{equation*}
a_i = \left\{ \begin{array}{ll} 
(-1)^{\frac{i + N - 2}{2}} & \quad \text{if\ } i \text{ is odd,} \\
0 & \quad \text{if\ } i \text{ is even.} 
\end{array} \right.
\end{equation*}

We have the following
\begin{lm} \label{lemma}
The characteristic polynomial $\chi_P(\lambda)$ of the matrix $M_P$ from (\ref{naseMP}) is equal to
%\begin{equation}
%\frac{(k-1)^{N-1} }{2^{N-1}}U_{\frac{N-1}{2}} \left(T_2 \left( 1 - \frac{\lambda}{k-1} \right) \right).
%\end{equation}
%where $x = 2 - \frac{2 \lambda}{k-1}.$
%Moreover, 
\begin{equation}
\frac{1}{{2^{N-1}}}(k-1)^{N-1}U_{\frac{N-1}{2}}\left(T_2\left(\frac{x}{2}\right)\right),
\end{equation}
where $x=2 - \frac{2\lambda}{k-1}$
\end{lm}

\begin{proof}
Clearly, the characteristic polynomial $\chi_P$ of the matrix $M_P$ in (\ref{naseMP}) is given by the determinant of the matrix
\begin{equation}
M_P - \lambda E= \frac12(k-1) \left(
\begin{array}{cccccc} 
x & -1 & 0 & \ldots & 0 & b_1 \\
-1 & x & -1 & \ldots & 0 & b_2 \\
0 & -1 & x & \ddots & \vdots & \vdots \\
0 & 0 & -1 &  \ddots & -1 & b_{N-3} \\
\vdots & \vdots & \vdots & \ddots & x & b_{N-2} \\
0 & 0 & 0 & \ldots & -1 & x 
\end{array}
\right)
\end{equation}
with $x = 2 - \frac{2\lambda}{k-1},$ $b_i = a_i$ for $i = 1,2, \ldots, N-3,$ $b_{N-2} = -2.$ The polynomial $\chi_P(\lambda) = \det{(M_P - \lambda E)}$ can be computed by the Laplace expansion along the $(N-1)$st column:
\begin{eqnarray}
\chi_P(\lambda) & = & \frac{1}{2^{N-1}}(k-1)^{N-1} \sum_{k=1}^{N-1} (-1)^{N-1+k} b_k (-1)^{N-1-k} f_{k-1} \nonumber \\
			 & = & \frac{1}{2^{N-1}}(k-1)^{N-1} \sum_{k=1}^{N-1} b_k f_{k-1}, \label{chiprvni}
\end{eqnarray}
where $f_{j}$ denotes the determinant of the $j \times j$ sub-matrix
\begin{equation}
\left( \begin{array}{cccccc} 
x & -1 & 0 & 0 & \ldots & 0 \\
-1 & x & -1 & 0 & \ldots & 0 \\
0 & -1 & x & -1 & \ldots & 0 \\
0 & 0 & -1 & x & \ddots & \vdots \\
\vdots & \vdots & \vdots & \ddots & \ddots & -1 \\
0 & 0 & 0 & \ldots & -1 & x 
\end{array} \right)
\end{equation}
for $j= 0, 1, \ldots, N-2$ with $f_0 = 1.$ This is a Toeplitz tridiagonal matrix whose determinant is the so called continuant with both the minor diagonals consisting solely of $(-1)$'s and with the diagonal consisting of $x$'s. It is a well known fact (cf.~ \cite{muir}) that the continuant satisfies the following three term recurrence relation
\begin{equation}
f_{j+1} = x f_{j} - f_{j-1}
\end{equation}
which in our case moreover satisfies the initial conditions $f_0 = 1$ and $f_{1} = x.$ This is exactly the recurrence relation satisfied by $U_j(\frac{x}{2})$'s from (\ref{chebyshev2}). Thus $\chi_P(\lambda)$ is equal to the following linear combination of Chebyshev polynomials of the second kind:
\begin{eqnarray}
\chi_P(\lambda) & = & \frac{1}{2^{N-1}} (k-1)^{N-1} \sum_{k=1}^{N-1} b_k U_{k-1}\left(\frac{x}{2}\right) \label{chidruhe} \\
			 & = & \frac{1}{2^{N-1}} (k-1)^{N-1} \left(\sum_{k=0}^{\frac{N-5}{2}} (-1)^{\frac{2k+N-1}{2}} U_{2k}\left(\frac{x}{2} \right) - 2U_{N-3}\left(\frac{x}{2} \right) + x U_{N-2}\left(\frac{x}{2} \right) \right) \nonumber \\
			 & = & \frac{1}{2^{N-1}} (k-1)^{N-1} \sum_{k=0}^{\frac{N-1}{2}} (-1)^{\frac{2k+N-1}{2}} U_{2k}\left(\frac{x}{2} \right), \label{chitreti}
\end{eqnarray}
the first equality follows upon substituting $U_k \left(\frac{x}{2} \right)$ for $f_k$ into (\ref{chiprvni}), the second one results from plugging the corresponding $b_k$'s into (\ref{chidruhe}) (the resulting sum with the upper index $\frac{N-5}{2}$ is to be interpreted as void for $N=3$) and the last equality is just a consequence of the definition (\ref{chebyshev2}). At this point we split the proof in two parts depending on whether the number of summands of the sum from (\ref{chitreti}) is even (the first case) or odd (the second case).

\textit{First case.} This case is equivalent to the fact that $N \equiv 3\!\mod{4}.$ Then the sum from (\ref{chitreti}) can be rewritten in the form
\begin{eqnarray}
\sum_{k=0}^{\frac{N-1}{2}} (-1)^{\frac{2k+N-1}{2}} U_{2k}\left(\frac{x}{2} \right) & = & U_{N-1}\left(\frac{x}{2} \right) - U_{N-3}\left(\frac{x}{2} \right) + U_{N-5}\left(\frac{x}{2} \right) - \ldots + U_2\left(\frac{x}{2} \right) - U_0\left(\frac{x}{2} \right) \nonumber \\
& = & 2\left(T_{N-1}\left(\frac{x}{2} \right) + T_{N-5}\left(\frac{x}{2} \right) + \ldots + T_2\left(\frac{x}{2} \right) \right), \label{Tvyraz}
\end{eqnarray}
where we used (\ref{vlastnost1}) in the second equality. Since $N$ is odd, thus $N-1$ is even, we have $N-1 = 2r$ for certain $r \in \mathbb{N},$ $r$ odd, and we can use (\ref{nesting}) to recast the expression (\ref{Tvyraz}) in the following way:
\begin{eqnarray*}
\lefteqn{2\left(T_{N-1}\left(\frac{x}{2} \right) + T_{N-5}\left(\frac{x}{2} \right) + \ldots + T_2\left(\frac{x}{2} \right) \right) = 2 \left(T_r \left(T_2\left(\frac{x}{2} \right)\right) + T_{r-2}\left(T_2\left(\frac{x}{2} \right)\right) + \ldots + T_1\left(T_2\left(\frac{x}{2} \right)\right) \right)} \\
& = & U_r \left(T_2\left(\frac{x}{2} \right)\right) = U_{\frac{N-1}{2}} \left(T_2\left(\frac{x}{2} \right)\right). \qquad \qquad \qquad \qquad \qquad \qquad \qquad \qquad \qquad \qquad \qquad \qquad \qquad \qquad
\end{eqnarray*}

\textit{Second case.} This case is in turn equivalent to the fact that $N \equiv 1\! \mod{4}.$ The sum from (\ref{chitreti}) now takes the form
\begin{eqnarray}
\sum_{k=0}^{\frac{N-1}{2}} (-1)^{\frac{2k+N-1}{2}} U_{2k}\left(\frac{x}{2} \right) & = & U_{N-1}\left(\frac{x}{2} \right) - U_{N-3}\left(\frac{x}{2} \right) + U_{N-5}\left(\frac{x}{2} \right) - \ldots - U_2\left(\frac{x}{2} \right) + U_0\left(\frac{x}{2} \right) \nonumber \\
& = & 2\left(T_{N-1}\left(\frac{x}{2} \right) + T_{N-5}\left(\frac{x}{2} \right) + \ldots + T_4\left(\frac{x}{2} \right) \right) + U_0\left(\frac{x}{2} \right), \label{Tvyraz1}
\end{eqnarray}
where again the property (\ref{vlastnost1}) was employed. Since $N$ is odd, thus $N-1$ is even, so that $N-1 = 2r$ for certain $r \in \mathbb{N},$ $r \geq 2,$ $r$ even, and using the property (\ref{nesting}) we obtain
\begin{eqnarray*}
\lefteqn{2\left(T_{N-1}\left(\frac{x}{2} \right) + T_{N-5}\left(\frac{x}{2} \right) + \ldots + T_4\left(\frac{x}{2} \right) \right) + U_0\left(\frac{x}{2} \right)} \\
& = & 2 \left(T_r \left(T_2\left(\frac{x}{2} \right)\right) + T_{r-2}\left(T_2\left(\frac{x}{2} \right)\right) + \ldots + T_2\left(T_2\left(\frac{x}{2} \right)\right) \right) + U_0 \left(\frac{x}{2} \right) \\
& = & U_r \left(T_2\left(\frac{x}{2} \right)\right) - 2 T_0 \left(T_2 \left(\frac{x}{2} \right) \right) + 1 + U_0 \left(\frac{x}{2} \right) \\
& = & U_{r} \left(T_2\left(\frac{x}{2}\right)\right) - 2 + 1 + 1 \\
& = & U_{\frac{N-1}{2}} \left(T_2\left(\frac{x}{2}\right) \right),
\end{eqnarray*}
where in the second and third equality we respectively used the property (\ref{vlastnost3}) and the fact that $T_0(x) = 1$ and $U_0(x) = 1$ for all $x.$ 
\end{proof}
\begin{pz}\label{pz1}
Note that the transformation property (\ref{transform}) implies that, upon writing 
\begin{equation*}
U_N\left(\frac{x}{2} \right) = U_{2 \frac{N+1}{2} - 1} \left(\frac{x}{2} \right) = U_{\frac{N-1}{2}} \left(T_2\left(\frac{x}{2} \right) \right) U_1\left(\frac{x}{2} \right)=U_{\frac{N-1}{2}} \left(T_2\left(\frac{x}{2} \right) \right)x,\label{formulka}
\end{equation*}
the characteristic polynomial $\chi_P(\lambda)$ can be written as
$$\chi_P(\lambda)=\frac{1}{{2^{N-1}}}(k-1)^{N-1}\frac{U_N\left(\frac{x}{2}\right)}{x}=\frac{(k-1)^{N}}{{2^{N}}(k-1-\lambda)}U_N\left(1-\frac{\lambda}{k-1}\right).$$
\end{pz}

To be able to apply Theorem \ref{Costin theorem}, we have to show that the matrix $M_P$ corresponding to $P= - \nabla V_{N,k}(\mathbf{q})$ has $N-1$ distinct eigenvalues. This task is achieved in the following result.
\begin{prp} \label{koreny}
The matrix $M_P$ from (\ref{naseMP}) has precisely $N-1$ distinct eigenvalues for every odd $N \geq 3$ and every even $k\geq 4.$
\end{prp}

\begin{proof}
%This result follows easily from the transformation property (\ref{transform}) upon writing 
%\begin{equation}
%U_N\left(\frac{x}{2} \right) = U_{2 \frac{N+1}{2} - 1} \left(\frac{x}{2} \right) = U_{\frac{N-1}{2}} \left(T_2\left(\frac{x}{2} \right) \right) U_1\left(\frac{x}{2} \right). \label{formulka}
%\end{equation}

The claim follows from Lemma \ref{lemma} and the last remark since $U_N(x/2)$ has precisely $N$ distinct roots, whence we can infer that the polynomial
%$$U_{\frac{N-1}{2}} \left(T_2\left(\frac{x}{2} \right) \right)$$
$$\frac{U_N\left(\frac{x}{2}\right)}{x}$$
has precisely $N-1$ distinct roots.
\end{proof}

The next matter of interest that we heavily rely on later is the number of (ir)rational eigenvalues. This is treated in

\begin{prp} \label{racionalni}
The only possible rational eigenvalues of the matrix $M_P$ from (\ref{naseMP}) are equal to $\frac12(k-1)$ or $\frac32(k-1).$ Moreover, all the eigenvalues belong to the interval $(0, 2k-2)$ for every odd $N \geq 3$ and every even $k \geq 4.$
\end{prp}

\begin{proof}
First we note that the characteristic polynomial $\chi_P(\lambda)$ of $M_P$ can be written in the form
\begin{equation}
\chi_P(\lambda) = \frac{1}{{2^{N-1}}}(k-1)^{N-1}\sum_{k=0}^{\frac{N-1}{2}} (-1)^{k} {{N-k}\choose{k}}x^{N-2k-1}, \quad x=2-\frac{2\lambda}{k-1}.
%{\frac{N+1}{2} + k \choose 2k + 1} x^{2k}, \quad x = 2 - \frac{2\lambda}{k-1}.
\end{equation}
In particular, if we normalize the polynomial $\chi_P$ as a polynomial in $x$ (which does not affect the roots anyway), the leading coefficient of $\chi_P$ as a polynomial of the variable $x$ is equal to $1$ and the constant term of the same polynomial (of the variable $x$) is $(-1)^{\frac{N-1}{2}} \frac{N+1}{2}$ which is always nonzero. Therefore, according to the standard result on the rational roots of a polynomial with integer coefficients, we conclude that the only rational roots of the polynomial $\chi_P$ in $x$ are the divisors of $\frac{N+1}{2}$ (which is always integer, since $N$ is odd). In fact, from Remark \ref{pz1} we can see that all the roots of $\chi_P$ in $x$ are also the roots of the polynomial $U_N(x/2)$ and therefore they must belong to the interval $(-2,2).$ Hence the only rational roots of $\chi_P$ as a polynomial of the variable $x$ are the numbers $-1,0,1.$ Using Remark \ref{pz1} again, we can easily see that $0$ is in fact not a root of $\chi_P$ as a polynomial in $x,$ since $x=0$ is the only root of $U_1(x/2) = x$ and therefore, all roots of $U_{N}(x/2)$ being simple, it cannot be the root of $U_{\frac{N-1}{2}}(T_2(x/2)).$ Since the number $x = 2 - \frac{2\lambda}{k-1}$ is rational if and only if the number $\lambda$ is rational, we can infer that the only rational roots of the polynomial $\chi_P(\lambda)$ are those numbers $\lambda$ satisfying the relations $\pm1 = 2 - (2\lambda)/(k-1),$ which means that $\lambda = (k-1)/2$ or $\lambda = 3(k-1)/2$ as claimed.

The fact that all eigenvalues $\lambda$ belong to the interval $(0,2k-2)$ clearly follows as $x = 2- (2\lambda)/(k-1)$ and $x \in (-2,2)$ in view of the considerations made above.
\end{proof}

\begin{prp} \label{iracionalni}
Suppose $N \geq 3$ is odd. Then
\begin{enumerate}
\item{For an arbitrary $N$ such that $N \equiv 1$ or $N \equiv 3\!\! \mod 6$ and for every even $k\geq 4$, all the numbers $\nu_m$ in (\ref{nyem}), $m=1,\dots,N-1,$ corresponding to the Hamiltonian system given by the Hamiltonian (\ref{system}), are irrational numbers.}
\item{For an arbitrary $N$ such that $N \equiv 5\!\! \mod 6$ and for every even $k\geq 4$, at least $N-3$ numbers among the numbers $\nu_m$ in (\ref{nyem}), $m=1,\dots,N-1,$ corresponding to the Hamiltonian system given by the Hamiltonian (\ref{system}), are irrational numbers.}
\end{enumerate}
\end{prp}

\begin{proof}
First note that all the numbers $\nu_m$ are in fact real (this follows from the previous proposition). Then, clearly, if some of the eigenvalues $\lambda_m$ of the matrix $M_P$ is irrational for certain $m = 1, \ldots, N-1$, then the corresponding number $\nu_m$ is also irrational. This means that the only possibility for $\nu_m$ to be rational is when the corresponding number $\lambda_m$ is rational. This situation occurs, due to Proposition \ref{racionalni}, only if $\lambda = (k-1)/2$ or $\lambda = 3(k-1)/2.$ The fact that these two values of $\lambda$ are roots of the characteristic polynomial $\chi_P$ is equivalent to the fact that
\begin{equation*}
U_{\frac{N-1}{2}} \left(T_2 \left(\pm \frac12 \right) \right) = U_{\frac{N-1}{2}} \left(- \frac12 \right)=0.
\end{equation*}
This last condition is in turn satisfied only if there is an integer $j = 1, \ldots, \frac{N-1}{2}$ such that
\begin{equation*}
- \frac12 = \cos \left( \frac{2 j \pi}{N+1} \right),
\end{equation*}
according to formula (\ref{koreny}). This happens precisely when $N$ is such that $N+1$ is divisible by $3$ which means, since $N$ is odd, that $N \equiv 5 \!\! \mod{6}.$

%The corresponding numbers $\nu$ are then, in the respective order, equal to $\nu = (\sqrt{5k^2 - 8k + 4})/4$ or $\nu = (\sqrt{13k^2 - 16k + 4})/4.$ However, neither of these two numbers is rational, since neither $5k^2 - 8k + 4$ nor $13k^2 - 16k + 4$ is a perfect square (because neither $5$ nor $13$ is). %can be written as a square with rational coefficients.
\end{proof}

In view of Proposition \ref{iracionalni}, we have thus proved the following theorem:

\begin{vt}
For an arbitrary $N \geq 3$ such that $N \equiv 1$ or $N \equiv 3 \!\! \mod 6$ and for arbitrary even $k\geq 4,$ the Hamiltonian system with $N$ degrees of freedom, given by the Hamiltonian (\ref{system}) with the potential (\ref{potential}) 
%\begin{equation}
%H = \frac{1}{2}%\sum_{i=1}^{N} p_{i}^{2}
%\mathbf{p}^2 + V_{N,k}(\mathbf{q}),
%\end{equation}
%where 
%\begin{equation}
%\frac12\mathbf{p}^2 = \sum_{i=1}^{N} p_{i}^{2}, \quad V_{N,k}(\mathbf{q}) = \sum_{i=0}^{N} (q_i - q_{i+1})^k, \ \quad \text{with\ } q_0 = q_{n+1} = 0, 
%\end{equation}
does not admit an additional first integral meromorphic in a complex neighbourhood of the linear manifold $q_m=a_m q_N$, $p_m=a_m p_N, m=1, \dots, N-1,$ where  $a_m=(-1)^{\frac{m+N-2}{2}}$ for $m$ odd, and $a_m=0$ for $m$ even.
For an arbitrary $N$ such that $N \equiv 5 \!\! \mod 6,$ the Hamiltonian system with $N$ degrees of freedom, given by the Hamiltonian (\ref{system}) with the potential (\ref{potential}) has at most four additional first integrals meromorphic in a complex neighbourhood of the linear manifold $q_m=a_m q_N$, $p_m=a_m p_N, m=1, \dots, N-1,$ where  $a_m=(-1)^{\frac{m+N-2}{2}}$ for $m$ odd, and $a_m=0$ for $m$ even.
\end{vt}

\begin{proof}
The claim follows from Proposition \ref{iracionalni} and Theorem \ref{Costin theorem}.
\end{proof}

This last theorem shows that Yoshida's Conjecture is true for every $N$ such that $N \equiv 1 \!\! \mod 6$ or $N \equiv 3 \!\! \mod 6$ with $k \geq 4$ an arbitrary even number. The case $N \equiv 5 \!\! \mod 6$ still leaves certain possibility for four additional global meromorphic first integrals to exist. However, it is clear that the only ${\nu}_m$'s that are possibly rational (and for which the criterion of Costin's is not decisive) are the numbers
\begin{equation*}
\nu = \frac{\sqrt{5k^2 - 8k + 4}}{4}, \quad \text{or \ } \nu = \frac{\sqrt{13k^2 - 16k + 4}}{4}.
\end{equation*}
But these numbers are clearly rational if and only if the corresponding expressions under the square-root sign, $5 k^2 - 8k + 4$ and $13k^2 - 16k + 4,$ are perfect squares, respectively. Although the task to determine precisely those $k$'s for which this situation occurs is quite cumbersome (if not impossible), reducing our ability to prove Yoshida's Conjecture in its full generality to a seemingly simple number-theoretic problem, the computer-aided numerical experiments that we made suggest that at least among the first $10^8$ $k$'s the only even ones, for which the corresponding $\nu$'s are rational, are quite rare (there are in fact only $3$ of them in each case). Moreover, it turns out that in neither of these cases both $\nu$'s are rational simultaneously. This actually seizes the space left for possible additional meromorphic first integrals of the system in question from $4$ to $2$ within the range $k = 4, \ldots ,10^8.$
\vspace{1cm}
\begin{center}
\begin{tabular}{c||c|c}
\hline
 k & $\nu = \frac{\sqrt{5k^2 - 8k + 4}}{4}$ & $\nu = \frac{\sqrt{13k^2 - 16k +4}}{4}$ \\
 \hline
 16 & rational & irrational \\
 40 & irrational & rational \\
 760 & irrational & rational \\
 4896 & rational & irrational \\
 1576240 & rational & irrational \\
 66354520 & irrational & rational\\
\end{tabular}
\end{center}
\vspace{1cm}

Taking into account the results of Maciejewski, Przybylska and Yoshida \cite{yoshida2}, we can state the following

\begin{ds}
Yoshida's Conjecture holds true:
\begin{enumerate}
\item{for every $N \geq 3$ such that $N \equiv 1 \!\! \mod{6}$ or $N \equiv 3 \!\! \mod{6}$ and every even $k \geq 4;$}
\item{for $N=5$ and for arbitrary even $k \geq 4;$}
\item{for every $N \geq 11$ such that $N \equiv 5 \!\! \mod{6}$ and every even $k = 4, \ldots, 10^8,$ where $k \neq$ $16,$ $40,$ $760,$ $4896,$ $1576240,$ $66354520.$}
\end{enumerate}
For every $N \geq 11$ such that $N \equiv 5 \!\! \mod{6}$ with $k$ equal to any of the numbers $$16, 40, 760, 4896, 1576240, 66354520,$$ the Hamiltonian system with $N$ degrees of freedom given by the Hamiltonian (\ref{system}) with the potential (\ref{potential}) admits at most two additional first integrals that are meromorphic in a complex neighbourhood of the linear manifold $q_m=a_m q_N$, $p_m=a_m p_N, \ m=1, \dots, N-1,$ where  $a_m=(-1)^{\frac{m+N-2}{2}}$ for $m$ odd, and $a_m=0$ for $m$ even.
\end{ds}
\begin{pz}
Note that in the case of the system (\ref{system}) there is the following relation between the Yoshida's  numbers $\Delta_j$ from (\ref{yoshidasuff}) and the Costin's numbers $\nu_j$: Let $\lambda_1,\dots,\lambda_{N-1}$ denote the eigenvalues of the matrix (\ref{naseMP}). Then they are equal (according to Remark \ref{pz1}) to those $N-1$ numbers $\lambda_j$ from (\ref{yoshidasuff}) that are different from $k-1$. Obviously, $$\nu_j=\frac{k}{2}\Delta_j$$ for every $j=1,\dots,N-1$. Therefore, we can see that instead of independence of $\Delta_j$'s over $\mathbb{Q}$ we can require their irrationality as the sufficient condition for non-integrability of the system (\ref{system}), which can be sometimes (especially if some parameters are present in the system) easier to decide.
\end{pz}
\begin{pz}
In general, the fact that a complex Hamiltonian system does not admit an additional meromorphic first integral (locally or globally) does not imply that the system in question exhibits chaos in the real phase space. Therefore, it is interesting to consider a family (indexed by the parameters $N$ and $k$) of real versions of Hamiltonian systems (\ref{system}) and to perform further investigation in this direction. One of the tools commonly used in this area is the so-called Poincar\'{e} section method which rests on the study of the discrete dynamical system given by the Poincar\'{e} map corresponding to a periodic trajectory of the system and to a chosen hypersurface in the phase space (the Poincar\'{e} hypersurface) which is transverse to the trajectory. 
%Thus, the state space of this discrete dynamical system is one dimension smaller than the phase space of the original system, and the orbit of a point under the Poincar\'{e} map is the set of all intersection points of the trajectory with the Poincar\'{e} hypersurface. 
Since it is very difficult to find the Poincar\'{e} map in general, some numerical methods are often used in order to find the orbit of a given point under this map.
%Instead, the orbit of a point under the Poincar\'{e} map is often computed using some numerical methods. 
The orbit can be easily visualized in case of two-degree-of-freedom Hamiltonian systems given by time-independent Hamiltonians (one first integral reduces the dimension of the phase space to 3 so that the dimension of the Poincar\'{e} hypersurface is equal to 2) and a wealth of information on the dynamical behaviour of the original system can be often obtained from this visualization. 
However, we consider Hamiltonian systems with $N$ degrees of freedom with $N\geq 3$ in this paper, and, for this reason, there is no reasonable way how to visualize the Poincar\'{e} hypersurfaces (their dimension is equal to $2N-2$),
and hence it would require much more effort to obtain any information on the behavior of the original system, which is beyond the scope of the present article.%Hence, we defer this numerical investigation to a later paper.
\end{pz}
\section{Concluding remarks}
%We have proved that the the Hamiltonian system with $N$ degrees of freedom given by the Hamiltonian (\ref{system}) with the potential (\ref{potential}) where $N\geq 3$ is odd and $k\geq 4$ is even, has no meromorphic first integral other than $H$. 
To sum up, we have proved that, for an arbitrary even number $k\geq 4$, the system ($\ref{system1}$) with Hamiltonian (\ref{system}) does not admit an additional first integral meromorphic in a complex neighbourhood of the linear manifold $q_m=a_m q_N$, $p_m=a_m p_N, m=1, \dots, N-1,\ a_m=(-1)^{\frac{m+N-2}{2}}$ for $m$ odd,  $a_m=0$ for $m$ even, 
 for all $N\geq 3$ odd such that $N\equiv 1 \mathrm{\ mod\ } 6$ or $N\equiv 3 \mathrm{\ mod\ } 6$, and that the same system admits at most four additional first integrals meromorphic near the linear manifold mentioned above for $N \geq 11$ odd such that $N\equiv 5 \mathrm{\ mod\ } 6$ (the case $N = 5$ with $k \geq 4$ was settled by Maciejewski, Przybylska and Yoshida in \cite{yoshida2}). Moreover, computer-aided numerical experiments suggest that even in this case (i.e. $N \equiv 5 \mathrm{\ mod\ } 6,$ $N \geq 11$), the system has no additional meromorphic first integral for ``most" of the values of $k \geq 4$ even. To give a more decisive answer, we have to solve two different problems: first, find all $k \geq 4$ even, such that $5 k^2 - 8k + 4$ and $13k^2 - 16k + 4,$ are perfect squares, respectively, and second, even if for infinitely many $k$'s the corresponding values of the expressions $5 k^2 - 8k + 4$ and $13k^2 - 16k + 4$ are perfect squares, it could well happen that the original system does not admit an additional meromorphic first integral as well, the poly-Painlev\'{e} test being simply indecisive here, so that still different tests are needed. Both of these last two problems deserve to be studied separately.

Last but not least, a natural question arises: what happens if we consider different values of $N$ and $k$ than those considered above. It turns out that for $k = 0,1,2$ (with $N$ arbitrary) the corresponding system is linear and therefore integrable in any reasonable sense. For $N \geq 0$ even (with an arbitrary $k \geq 3$) we weren't able to find any nonzero point $\pmb{\alpha}$ in a closed form for which the procedure would work. For $N \geq 5$ odd that is congruent to $2\!\!\mod 3$ (with $k \geq 3$ odd) we found a point $\pmb{\alpha} = \left(\alpha, -\alpha, 0, \alpha, -\alpha, 0, \ldots, \alpha, -\alpha \right),$ where $\alpha = (k(1 - 2^{k-1}))^{1/(2-k)}.$ However, in this case we weren't able to say too much about the eigenvalues of the associated matrix $M_P$ in general. We made some numerical experiments for $N=5$ and $k\geq 3$ odd and it seems that for $k\geq 21$ all the four numbers $\nu_m$ could be irrational so that the system (\ref{system}) seems to possess no additional first integral meromorphic in a complex neighbourhood of the linear manifold $q_m=\frac{{\alpha}_m}{\alpha_N} q_N$, $p_m=\frac{{\alpha}_m}{\alpha_N}p_N, m=1, \dots, N-1$ in this case. In other cases of $N$ and $k$ remaining, we weren't again able to find the point $\pmb{\alpha}.$

The results used in this paper rest essentially on the so called poly-Painlev\'{e} method which was proposed at the beginning of 90's by M. Kruskal. This method is based on asymptotic expansions of the unknown solutions, an idea which is close to the so called Painlev\'{e} $\alpha$-method. It seems that Kruskal himself in fact conjectured that if dense branching of a truncated asymptotic series to a certain order occurs, the actual solutions also have a dense branching and therefore there are no continuous first integrals of the corresponding equation. Although this conjecture is still not proved in its full generality, there are some partial results where this method is rigorously justified \cite{costin1}, \cite{costin2}, \cite{CostinThesis}, \cite{costinkruskal}, one of which was used in our proof.  

%Since the poly-Painlev\'{e} method is in fact a nonlinear extension of the theory of Ziglin'{s} to higher orders and since, on the other hand, the theory of higher variational equations due to Morales-Ruiz et al.~ \cite{morales-ruiz5} may be viewed as an improvement on the theory of Ziglin's, it is a natural question whether there is any connection between the two.

\section*{Acknowledgements}
The authors thank Dr.~ A. Sergyeyev and Prof.~ I.
 S. Krasil'shchik for helpful suggestions. This research was supported by the institutional support for I\v{C}47813059. The work of Ji\v{r}\'{i} Jahn was supported by the grant SGS 1/2013. The work of Ji\v{r}ina Jahnov\'{a} was in part supported by the fellowship from the Moravian-Silesian region. The authors also thank the referees for valuable comments.
%The author thanks Dr. A. Sergyeyev for stimulating discussions.
%????This research was supported by the Silesian university in Opava????, and by the fellowship
%from the Moravian--Silesian region.
\looseness=-1


\begin{thebibliography}{99}
\bibitem{arnold} Arnold, V. I.: \textit{Mathematical methods of classical mechanics}, Springer New York, 1978.
\bibitem{costin1} Costin, R. D.: \textit{Integrability Properties of a Generalized Lam\'{e} Equation: Applications to the H\'{e}non-Heiles System}, Methods Appl. Anal. 4 (1997), 113-123.
\bibitem{costin2} Costin, R. D.: \textit{Integrability properties of nonlinearly perturbed Euler equations}, Nonlinearity 10 (1997), 905--924.
\bibitem{CostinThesis} Costin, R. D.: \textit{Applications of the poly-Painlev\'{e} test}, PhD Thesis, Rutgers University, 1997.
\bibitem{costinkruskal} Costin, R. D., Kruskal, M. D.: \textit{Nonintegrability criteria for a class of differential equations with two regular singular points}, Nonlinearity 16 (2003), 1295--1317.
\bibitem{kruskal4} Kruskal, M. D.: \textit{Modified Painlev\'{e} test for analytic integrability}, Structure, Coherence and Chaos in Dynamical Systems (eds. Christiansen, P. L., Parmentier, R. D.), Manchester University Press, 1989.
\bibitem{kruskal1} Kruskal, M. D., Ramani, A., Grammaticos B.: \textit{Singularity analysis and its relation to complete, partial and non-integrability}, Partially Intergrable Evolution Equations in Physics, Springer Netherlands (1990), 321-372.
\bibitem{kruskal2} Kruskal, M. D., Clarkson, P. A.: \textit{The Painlev\'{e}-Kowalevski and poly-Painlev\'{e} tests for integrability}, Studies in applied mathematics 86 (1992), 87-165.
\bibitem{kruskal3} Kruskal, M. D.: \textit{``Completeness" of the Painlev\'{e} Test--General Considerations--Open Problems}, The Painlev\'{e} property, Springer New York (1999), 789-804.
\bibitem{maciejewski} Maciejewski, A. J., Przybylska M.: \textit{Partial integrability of Hamiltonian systems with homogeneous potential}, Regular and Chaotic Dynamics 15 (2010), 551--563.
\bibitem{yoshida2} Maciejewski, A. J., Przybylska, M., Yoshida, H.: \textit{Necessary conditions for the existence of additional first integrals for Hamiltonian systems with homogeneous potential.} Nonlinearity 25 (2012), 255-277.
\bibitem{mason} Mason, J. C., Handscomb, D. C.: \textit{Chebyshev polynomials}, CRC Press (2010).
\bibitem{morales-ruiz4} Morales-Ruiz, J. J.: \textit{Differential Galois theory and non-integrability of Hamiltonian Systems}, Birkh\"{a}user, 1999.
\bibitem{morales-ruiz1} Morales-Ruiz, J. J., Ramis, J.-P.: \textit{Galoisian obstructions to integrability of Hamiltonia systems}, Methods Appl. Anal. 8 (2001), 33--96.
\bibitem{morales-ruiz2} Morales-Ruiz, J. J., Ramis, J.-P.: \textit{A note on the non-integrability of some Hamiltonian systems with a homogeneous potential}, Methods Appl. Anal. 8 (2001), 113--120.
\bibitem{morales-ruiz5} Morales-Ruiz, J. J., Ramis, J.-P., Sim\'{o}, C.: \textit{Integrability of Hamiltonian systems and differential Galois groups of higher variational equations}, Ann. Scient. \'{E}c. Norm. Sup. 40 (2007), 845-884.
\bibitem{morales-ruiz3} Morales-Ruiz, J. J., Simon, S.: \textit{On the meromorphic non-integrability of some $N$-body problems}, Discrete Contin. Dyn. Syst. 24 (2009), 1225--1273.
\bibitem{muir} Muir, T.: \textit{A treatise on the theory of determinants}, Dover Publications, 1960.
\bibitem{snyder} Snyder, M. A.: \textit{Chebyshev methods in numerical approximation}, Englewood Cliffs, NJ: Prentice-Hall  (1966).
\bibitem{ramani} Ramani, A., Grammaticos, B., Bountis, T.: \textit{The Painlev\'{e} property and singularity analysis of integrable and non-integrable systems}, Physics Reports 180 (1989), 159--245.
%\bibitem{umeno} Umeno, K., \textit{Galois extensions in Kowalevski exponents and nonintegrability of nonlinear lattices}, Physics Letters A 190, no. 1 (1994), 85--89.
\bibitem{yoshida3} Yoshida, H.: \textit{A criterion for the non-existence of an additional analytic integral in Hamiltonian systems with homogeneous potential}, Physica D 29 (1987), 128-142.
\bibitem{yoshida1} Yoshida, H.: \textit{A criterion for the non-existence of an additional analytic integral in Hamiltonian systems with $n$ degrees of freedom}, Physics letters A 141 (1989), 108-112.
\bibitem{yoshimura} Yoshimura, K.: \textit{Non-integrability of homogeneous nonlinear lattices}, RIMS K\^{o}ky\^{u}roku 1543 (2007), 220--226 (in Japanese)
\bibitem{ziglin1}Ziglin, S. L.: \textit{Branching of solutions and non-existence of first integrals in Hamiltonian mechanics I}, Funct. Anal. Appl. 16 (1982), 181--189.
\bibitem{ziglin2} Ziglin, S. L.: \textit{Branching of solutions and non-existence of first integrals in Hamiltonian mechanics II}, Funct. Anal. Appl. 17 (1983), 6--17.
\end{thebibliography}
\end{document}